\documentclass[pra,twocolumn, showpacs,preprintnumbers,amsmath,amssymb,nofootinbib]{revtex4}

\usepackage{graphicx, color, graphpap}
\usepackage{amsmath}
\usepackage{amssymb}
\usepackage{amsthm}
\usepackage{multirow}
\usepackage{hyperref}

\long\def\ca#1\cb{} 

\ca
 \def\outl#1{}  \def\xa{} \def\xb{}  
\cb

\ca
 \def\outl#1{\par{\medskip\noindent\hspace*{.5cm}\bf
      \mathversion{bold}#1\mathversion{normal}\smallskip} }
 \long\def\xa#1\xb{}
 
\cb

 \def\outl#1{\par{\medskip\noindent\hspace*{.5cm}\bf
      \mathversion{bold}#1\mathversion{normal}\smallskip} }
 \def\xa{} \def\xb{}  

\newcommand{\ket}[1]{|#1\rangle}               
\newcommand{\colo}{\,\hbox{:}\,}              
\newcommand{\bra}[1]{\langle #1|}              
\newcommand{\dya}[1]{\ket{#1}\bra{#1}}
\newcommand{\ip}[2]{\langle #1|#2\rangle}      
\newcommand{\HS}{\text{HS}}

\newcommand{\EC}{\mathcal{E}}

\newcommand{\HC}{\mathcal{H}}

\newcommand{\Tr}{{\rm Tr}}

\renewcommand{\geq}{\geqslant}
\renewcommand{\leq}{\leqslant}

\newcommand{\ot}{\otimes}
\newcommand{\ad}{^\dagger}

\newcommand{\al}{\alpha }
\newcommand{\bt}{\beta }

\newcommand{\dl}{\delta }

\newcommand{\ep}{\epsilon}

\newcommand{\lm}{\lambda }

\newcommand{\sg}{\sigma }

\newcommand{\Om}{\Omega }

\newtheoremstyle{example}{\topsep}{\topsep}%
{}
{}
{\bfseries}
{.}
{   }
{\thmname{#1}\thmnumber{ #2}}
\theoremstyle{example}

\newtheorem{theorem}{Theorem}
\newtheorem{proposition}[theorem]{Proposition}
\newtheorem{definition}{Definition}
\newtheorem{lemma}[theorem]{Lemma}

\begin{document}

\title{Non-negative discord strengthens the subadditivity of quantum entropy functions}
\author{Patrick J. Coles}
\affiliation{Department of Physics, Carnegie Mellon University, Pittsburgh,
Pennsylvania 15213, USA}
\date{Version of \today}


\begin{abstract}
The definition of quantum discord is generalized to allow for any concave entropy function, and its non-negativity strengthens the subadditivity condition for that entropy function. In a sense, this condition is \emph{intermediate} in strength between subadditivity and strong subadditivity, hence called firm subadditivity, allowing one to further classify entropy functions based on whether they satisfy this intermediate condition. It is proven that the quadratic entropy $1-\Tr(\rho^2)$ satisfies the firm subadditivity condition, whereas some other subadditive Tsallis entropies are not firmly subadditive.
\end{abstract}

\maketitle

\section{Introduction}\label{sct1}

The strong subadditivity is ``one of the most important and useful results in quantum information theory" (p. 519 of \cite{NieChu00}). This property of the von Neumann entropy leads to the intuitive notion that information does not increase upon discarding a subsystem or upon quantum operations, and it is helpful in proving certain entropic uncertainty relations \cite{RenesBoileau,ColesEtAl}. It is unfortunate that the von Neumann entropy $S(\rho)=-\Tr (\rho \log \rho)$ is the only entropy known to have this property, as it is not the easiest entropy to calculate due to the matrix logarithm. For example the quadratic (often called linear) entropy $S_Q(\rho)=1-\Tr \rho^2$ simply involves squaring the density operator rather than having to diagonalize it, making $S_Q$ easier to work with. 

Using the strong subadditivity property of $S$, one can show that the quantum mutual information between two systems $a$ and $b$ upper-bounds the Holevo quantity between $a$ and $b$, e.g.\ Lemma 1 of \cite{ColesEtAl}, or for earlier formulations that state this result as the non-negativity of the quantum discord, see Proposition 2 of \cite{OllZur01} or Theorem 1 of \cite{Datta2010}. This result is interesting in that it is purely quantum-mechanical with no classical analog, since the quantum mutual information equals (is the same thing as) the Holevo quantity for classically-correlated systems \cite{OllZur01}. We find it interesting enough to give it a name, \emph{firm subadditivity} (FSA), motivated by the fact that it is stronger than the subadditivity (SA) property (see below), yet it can be proven using the strong subadditivity (SSA) property, so ``firm" seems like an appropriate intermediate adjective. 

In this article, we show that $S$ is not the only entropy function that has the FSA property. Indeed the quadratic entropy $S_Q$ is \emph{firmly subadditive}. On the other hand, numerical counterexamples for particular Tsallis entropies (defined below) show that there are some entropies that are subadditive but not \emph{firmly} subadditive. Thus, the FSA property provides a new means with which to classify quantum entropy functions.

The remainder of this article is organized as follows. Section~\ref{sct2} introduces the notation, e.g.\ defining a generalized quantum discord for any concave entropy function, and gives further background information. Theorem~\ref{thm2} of Section~\ref{sct3} gives four mathematically-equivalent ways of stating the FSA property. For comparison, Theorem~\ref{thm3} of Section~\ref{sct4} gives four conditions for $S$ that are stronger than those in Theorem~\ref{thm2}. The proof that $S_Q$ satisfies FSA is given in Section~\ref{sct5}, followed by a discussion of the Tsallis entropies in Section~\ref{sct6}. Some remarks on a condition, related to zero discord, that generalizes the additivity condition are found in Section~\ref{sct7}, with concluding remarks in Section~\ref{sct8}.

\section{Notation and background}\label{sct2}

Consider the following quantum entropy functions,
\begin{align}
S(\rho)&=-\Tr (\rho \log \rho), \notag\\
S_R(\rho)&=\frac{1}{1-q} \log \Tr (\rho^q),\quad 0 < q\leq 1,\notag\\
S_T(\rho)&=\frac{1}{1-q} [\Tr (\rho^q)-1], \quad 0 < q,\notag\\
S_Q(\rho)&=1-\Tr (\rho^2),
\label{eqn1}
\end{align} 
respectively the von Neumann, Renyi, Tsallis, and quadratic (often called linear) entropies, where $S_R$ and $S_T$ are equal to $S$ in the limit $q=1$, and $S_T$ interpolates between $S$ and $S_Q$ as $q$ goes from 1 to 2. All of these entropies are concave, $S_K(\sum_j p_j\rho_j)\geq \sum_j p_j S_K(\rho_j)$ for $0<p_j<1$ and $\sum p_j=1$, for the ranges of $q$ specified in \eqref{eqn1}. Here, and in what follows we use the notation $S_K$ for any of the entropies in \eqref{eqn1}, dropping the $K$ subscript when specifically referring to von Neumann $S$; or if one wishes think of $S_K$ as a general entropy function \cite{HuYe}, assumed to depend only on the non-zero spectrum of its input and assumed to be concave.

Given a bipartite quantum system with a density operator $\rho_{ab}$, partial traces $\rho_a$ and $\rho_b$, one can define a generalized quantum mutual information:
\begin{equation}
\label{eqn2}
S_K(a\colo b)=S_K(\rho_a)+S_K(\rho_b)-S_K(\rho_{ab}),
\end{equation}
which is non-negative for any $\rho_{ab}$ if $S_K$ is subadditive (SA):
\begin{equation}
\label{eqn3}
S_K(\rho_a)+S_K(\rho_b)\geq S_K(\rho_{ab}),\quad S_K(a\colo b)\geq 0.
\end{equation}
Roughly speaking, thinking of entropy as the absence of information, \eqref{eqn3} says that one gains (never loses) information by obtaining access to the joint system $ab$, relative to the information one has from only local access to $a$ and $b$ separately. The entropies $S$, $S_Q$, and $S_T$ for $q\geqslant 1$ \cite{SubaddivityQ} are SA, but for $0< q< 1$ neither $S_T$ nor $S_R$ are SA. 

The SSA property of $S$ has various equivalent statements (p. 519 of \cite{NieChu00}), one being that for any tripartite state $\rho_{abc}$:
\begin{equation}
\label{eqn4}
S(a\colo bc)\geq S(a\colo b),
\end{equation}
the quantum mutual information does not increase upon discarding a subsystem.

While $S(a\colo b)$ is a very useful global or holistic measure of the correlation between $a$ and $b$, there is reason to believe that further insight can be gained from a measure that quantifies the presence of individual types of information \cite{Gri07} about $a$ in $b$. For example for a two-qubit system $ab$ described by density operator $\rho_{ab}=(\dya{00}+\dya{11})/2$, the $z$-information (0 or 1) about $a$ is perfectly contained in $b$ while the $x$-information (+ or $-$) about $a$ is completely absent from $b$. Precisely speaking, a type of information $P_a=\{P_{aj}\}$ about $a$ is a projective decomposition of the identity $I_a=\sum_j P_{aj}$ on $\HC_a$, as the $P_{aj}$ represent mutually exclusive properties analogous to a classical sample space. However this notion can be extended \cite{ColesEtAl} to POVMs, i.e.\ where the $P_{aj}$ are positive operators not necessarily projectors, provided the $P_{aj}$ are interpreted (physically) as orthogonal projectors on an enlarged Hilbert space, a Naimark extension. We employ the symbol $N_a=\{N_{aj}\}$ for a rank-1 POVM, i.e. a POVM in which all $N_{aj}$ are rank-1 operators. 

The example in the previous paragraph is an extreme case; in general one would like to quantify \emph{partial} information, for which one could use \cite{ColesEtAl} a Holevo quantity:
\begin{equation}
  \chi_K(\{p_j,\rho_j\})=S_K(\sum_jp_j\rho_j)-\sum_jp_jS_K(\rho_j),
\label{eqn5}
\end{equation}
which measures the distinguishability of the distinct $\rho_j$, and is non-negative and equal to zero if and only if all $\rho_j$ are identical, due to the (strict) concavity of all entropies in \eqref{eqn1}. As applied to the ensemble $\{p_j,\rho_{bj}\}$ generated by a POVM $P_a=\{P_{aj}\}$ on $a$ and a bipartite state $\rho_{ab}$,
\begin{equation}
 p_j\rho_{bj}=\Tr_a(P_{aj}\rho_{ab}),\quad p_j=\Tr(P_{aj}\rho_a),
\label{eqn6}
\end{equation}
the associated Holevo quantity
\begin{equation}
  \chi_K(P_a,b)=S_K(\rho_b)-S_K(\rho_b|P_a),
\label{eqn7}
\end{equation}
quantifies the amount of the $P_a$ type of information that is contained in $b$. In \eqref{eqn7} and later on we shall use $|P_a$ to denote conditioning on the result of measurement $P_a$, for example,
\begin{align}
  &S_K(\rho_b|P_a)=\sum p_jS_K(\rho_{bj}), \nonumber\\
   &S_K(\rho_a|P_a)=\sum p_jS_K(\rho_{aj}),\nonumber\\
\label{eqn8}  &S_K(a\colo b|P_a)= \sum p_j[S_K(\rho_{aj})+S_K(\rho_{bj})-S_K(\rho_{abj})],
\end{align}
where one uses the $p_j$ and $\rho_{bj}$ from \eqref{eqn6} here; $S_K(\rho_{aj})=S_K(K_j\rho_{a}K_j\ad)=S_K(\rho_{bcj})$ and $S_K(\rho_{abj})= S_K[(K_j\ot I)\rho_{ab}(K_j\ad\ot I)] =S_K(\rho_{cj})$ where $P_{aj}=K_j\ad K_j$, $c$ is a system that purifies $\rho_{ab}$, $\rho_{bcj}= \Tr_a(P_{aj}\rho_{abc})/p_j$, and $\rho_{cj}= \Tr_a(P_{aj}\rho_{ac})/p_j$.

We shall also find it useful to define a generalized quantum discord \cite{OllZur01}\footnote{Some authors take $\min_{P_a} \dl_K(P_a,b)$ as the discord, we shall do no such minimization in this article.}:
\begin{equation}
\label{eqn9}
\dl_K(P_a,b):= S_K(a\colo b)-\chi_K(P_a,b).
\end{equation}
A useful identity (an immediate consequence of Theorem 3 of \cite{ColesEtAl}), used in proving Theorem~\ref{thm2} below, is 
\begin{equation}
\label{eqn10}
\dl_K(N_a,c)= S_K(\rho_a)-\chi_K(N_a,b),
\end{equation}
where $N_a$ is a rank-1 POVM and $c$ is a system that purifies $\rho_{ab}$; so the quantity on the right-hand-side represents a discord if the purifying system is kept in mind.

\section{Firm subadditivity}\label{sct3}

We define the FSA condition as follows.
\begin{definition} 
\label{def1}
An entropy is \emph{firmly subadditive} or satisfies the \emph{firm subadditivity} condition if for all bipartite systems $ab$ (all dimensions), all $\rho_{ab}$, and all POVMs $P_a$ on $a$,
\begin{equation}
  S_K(a\colo b)\geq \chi_K(P_a,b).
\label{eqn11}
\end{equation}
\end{definition} 

\begin{proposition} 
\label{thm1}
For any concave entropy function $S_K$, the FSA condition implies and strengthens the SA condition; if $S_K$ is firmly subadditive, then it is also subadditive. \openbox
\end{proposition} 
The fact that \eqref{eqn11} is stronger than \eqref{eqn3} clearly follows from the fact that $\chi_K(P_a,b)\geq 0$ if $S_K$ is concave. Note that by symmetry one can obviously interchange the roles of $a$ and $b$ in \eqref{eqn11}, and also an alternative way of writing \eqref{eqn11} is
\begin{equation}
  \dl_K(P_a,b)\geq 0,
\label{eqn12}
\end{equation}
the quantum discord is non-negative. Furthermore, to ensure \eqref{eqn11} it is sufficient to check that $S_K(a\colo b)\geq \chi_K(w_a,b) $ for all orthonormal bases $w_a$, since by Naimark extension $P_a$ is equivalent to a projective measurement $P_A$ on a larger Hilbert space $\HC_A$, with $S_K(a\colo b)=S_K(A\colo b)$ and by the concavity of $S_K$, $\chi_K(P_A,b)\leq \chi_K(w_A,b)$ for some basis $w_A$ of $\HC_A$.

Conceptually, one can think of FSA in several ways. One can look at it as a statement that the ``quantum" correlations are bigger than the ``classical" correlations, but admittedly this is ambiguous language. Somewhat more precisely, the left-hand-side of \eqref{eqn11} is a global or holistic measure of the correlation between $a$ and $b$, whereas the right-hand-side quantifies a \emph{single type} of information about $a$ in $b$, so FSA says the global measure is bigger than the single-type measure. Alternatively, it says that the amount of information one gains by obtaining access to the joint system $ab$, relative to the local information in $a$ and $b$ separately, is not less than the amount that a single type of information about $a$ is present in $b$. 

The following theorem gives three other conditions that are mathematically equivalent to \eqref{eqn11}, for example, the Holevo quantity measuring the $P_a$ information in $b$ is bounded-above by the entropy of $\rho_a$ (part ii), and the Holevo quantity is non-increasing under the action of a quantum channel for an input ensemble of pure states (part iv). As written below, these inequalities show that, if $S_K$ is FSA, then this immediately implies some possibly useful bounds on the Holevo quantity $\chi_K$ associated with $S_K$.
\begin{theorem}
\label{thm2}
For any concave entropy function $S_K$ depending only on the (non-zero) spectrum of its input, the following four conditions are equivalent, each implying the other three:  

(i) The entropy $S_K$ is firmly subadditive; i.e.\ for any bipartite state $\rho_{ab}$ and any POVM $P_a$ on $a$,
\begin{equation}
\label{eqn13}
\chi_K(P_a,b)\leq S_K(a\colo b).
\end{equation}

(ii) For any bipartite state $\rho_{ab}$ with $\rho_a=\Tr_b(\rho_{ab})$ and any POVM $P_a$ on $a$, $S_K$ satisfies: 
\begin{equation}
\label{eqn14}
\chi_K(P_a,b)\leq S_K(\rho_a).
\end{equation}

(iii) For any tripartite \emph{pure} state $\rho_{abc}$ and \emph{rank-1} POVM $N_a$ on $a$, $\chi_K$ satisfies:
\begin{equation}
\label{eqn15}
\chi_K(N_a,b)\leq \chi_K(N_a,bc).
\end{equation}

(iv) For any quantum channel $\EC$ and any ensemble of pure states $\{p_j,\dya{\psi_j}\}$, $\chi_K$ satisfies:
\begin{equation}
\label{eqn16}
\chi_K(\{p_j,\EC(\dya{\psi_j})\}) \leq \chi_K(\{p_j,\dya{\psi_j}\}).
\end{equation}
\end{theorem}

\begin{proof}
We will show the equivalence of all four parts by showing (iv)$\Rightarrow$(iii)$\Rightarrow$(ii)$\Rightarrow$(i)$\Rightarrow$(iv). Starting from (iv), set the pure states $\dya{\psi_j}=\Tr_a[N_{aj}\rho_{abc}]/p_j= \rho_{bcj}$ to be the conditional density operators on $bc$ associated with the rank-1 POVM $N_a$ and pure state $\rho_{abc}$; this corresponds to setting $\chi_K(\{p_j,\dya{\psi_j}\})=\chi_K(N_a,bc)$. Now set $\EC$ in (iv) to be the channel that partial traces over system $c$, i.e.\ $\rho_{bj}=\EC(\rho_{bcj})=\Tr_c(\rho_{bcj})$, and the left-hand-side of \eqref{eqn16} becomes $\chi_K(N_a,b)$, proving (iii). By Theorem 2 of \cite{ColesEtAl}, $\chi_K(N_a,bc)=S_K(\rho_a)$, and by the concavity of $S_K$, $\chi_K(P_a,b) \leq \chi_K(N_a,b)$, where $P_a$, a general POVM, is a coarse-graining of some rank-1 POVM $N_a$, proving (ii). Now apply (ii) to $N_a$ and add $\chi_K(N_a,c)-\chi_K(N_a,b)=S_K(\rho_c)-S_K(\rho_b)$, see Theorem 3 of \cite{ColesEtAl}, to both sides giving $\chi_K(N_a,c)\leq S_K(a\colo c)$; again invoking concavity $\chi_K(P_a,c)\leq \chi_K(N_a,c)$ and interchange the labels $b$ and $c$ to prove (i). To prove (iv), introduce a reference system $a'$ of the same dimension as $a$; by Stinespring's dilation theorem any channel $\EC(\rho_{a'})=\Tr_c(V\rho_{a'}V\ad)$ from $a'$ to $c$ can be constructed by an isometry $V$ from $a'$ to $bc$ followed by a partial trace over $b$, and any pure-state ensemble $\{p_j,\dya{\psi_j}\}$ on the input $a'$ to $\EC$ can be produced by the action of some rank-1 POVM $N_a$ on $a$ with an appropriate choice of a partially-entangled pure state (pre-probability) $\ket{\Phi}$ on $aa'$, see Sec. IIB of \cite{ColesEtAl} for the explicit construction. Now apply (i) to $N_a$ and the tripartite pure state $\ket{\Om}=(I_a\ot V)\ket{\Phi}$ generated by acting $V$ on the $a'$ part of $\ket{\Phi}$, giving $\chi_K(N_a,c)\leq S_K(\rho_a)$ [after adding $\chi_K(N_a,c)-\chi_K(N_a,b)=S_K(\rho_c)-S_K(\rho_b)$ to both sides of \eqref{eqn13}.] But $\chi_K(N_a,c)=\chi_K(\{p_j,\EC(\dya{\psi_j})\})$ and $S_K(\rho_a)= \chi_K(N_a,a')=\chi_K(\{p_j,\dya{\psi_j}\})$, proving (iv).  
\end{proof}

\section{von Neumann entropy}\label{sct4}

For comparison, consider the following properties of $S$, each of which is stronger than its counterpart in Theorem~\ref{thm2}, and each of which can be proved \cite{SWW96} by invoking both the strong subadditivity \eqref{eqn4} and the generalized additivity [see \eqref{eqn27} below] of $S$. All four properties below are equivalent in the sense that if a general entropy $S_K$ satisfied one of them, it would automatically satisfy all of them.
\begin{theorem}
\label{thm3}
The von Neumann entropy $S$ and its associated Holevo quantity $\chi$ satisfy the following four \emph{equivalent} conditions:

(i) For any bipartite state $\rho_{ab}$ and any POVM $P_a$ on $a$:
\begin{equation}
\label{eqn17}
\chi(P_a,b)\leq S(a\colo b)-S(a\colo b|P_a).
\end{equation}

(ii) For any bipartite state $\rho_{ab}$ with $\rho_a=\Tr_b(\rho_{ab})$ and any POVM $P_a$ on $a$: 
\begin{equation}
\label{eqn18}
\chi(P_a,b)\leq S(\rho_a)-S(\rho_a|P_a).
\end{equation}

(iii) For any tripartite state $\rho_{abc}$ and any POVM $P_a$ on $a$:
\begin{equation}
\label{eqn19}
\chi(P_a,b) \leq \chi(P_a,bc).
\end{equation}

(iv) For any quantum channel $\EC$ and any ensemble $\{p_j,\rho_j\}$:
\begin{equation}
\label{eqn20}
\chi(\{p_j,\EC(\rho_j)\}) \leq \chi(\{p_j,\rho_j\}).
\end{equation}
\end{theorem}

\begin{proof}
Part (iii) was proven in \cite{SWW96} and we will show (iii)$\Rightarrow$(ii)$\Rightarrow$(i)$\Rightarrow$(iv)$\Rightarrow$(iii). For (ii), apply (iii) to a pure state $\rho_{abc}$, so $\chi(P_a,bc)=S(\rho_{bc})-\sum_j p_jS(\rho_{bcj})=S(\rho_{a})-\sum_j p_jS(\rho_{aj})$. For (i), let $c$ purify $\rho_{ab}$, add $\chi(P_a,c)-\chi(P_a,b)=S(\rho_c)-S(\rho_b)+S(\rho_b|P_a)-S(\rho_c|P_a)$ to both sides of \eqref{eqn18}, and interchange the labels $b$ and $c$ in the resulting inequality. For (i)$\Rightarrow$(iv), apply the same argument used in going from (i)$\Rightarrow$(iv) in the proof of Theorem~\ref{thm2}, replacing the pure-state ensemble with an arbitrary ensemble $\{p_j,\rho_j\}$ and $N_a$ with an arbitrary POVM $P_a$. Likewise for (iv)$\Rightarrow$(iii), let $\{p_j,\rho_j\}$ in (iv) be the ensemble generated by $p_j\rho_j=\Tr_a[P_{aj}\rho_{abc}]$, and set $\EC$ in (iv) to be the channel that partial traces over $c$.
\end{proof}

\section{Quadratic entropy}\label{sct5}

The proof that $S_Q$ is firmly subadditive relies on writing $\chi_Q$ in terms of the Hilbert-Schmidt distance:
\begin{equation}
\label{eqn21}
D_{\HS}(\rho,\sg)=\Tr[(\rho-\sg)^2]
\end{equation}
and using a property of $D_{\HS}$ given in the following lemma.

\begin{lemma}
\label{thm4}
Let $\EC$ be any quantum channel, let $\rho=\dya{\phi}$ and $\sg=\dya{\psi}$ be \emph{pure} states,

(i) For $q\geq 1$, with $|X|=(X\ad X)^{1/2}$,
\begin{equation}
\label{eqn22}
\Tr[| \rho-\sg|^q]\geq\Tr[|\EC(\rho)-\EC(\sg)|^q].
\end{equation}

(ii) Setting $q=2$ gives
\begin{equation}
\label{eqn23}
D_{\HS}(\rho,\sg)\geq D_{\HS}[\EC(\rho),\EC(\sg)].
\end{equation}
\end{lemma}
\begin{proof}
Write $\rho-\sg=\lm (\dya{\al}-\dya{\bt})$ where $\ip{\al}{\bt}=0$ and $0\leq \lm \leq 1$, and $\EC(\rho)-\EC(\sg)=\lm [\EC(\dya{\al})-\EC(\dya{\beta})]=\lm[A-B]$ where $A$ and $B$ are positive operators with orthogonal support $AB=0$, and $\Tr A=\Tr B \leq 1$. Then $\Tr[| \rho-\sg|^q]=2\lm^q\geq \lm^q [(\Tr A)^q+(\Tr B)^q]\geq \lm^q [\Tr (A^q)+\Tr (B^q)]= \lm^q \Tr |A-B|^q = \Tr[|\EC(\rho)-\EC(\sg)|^q]$.
\end{proof}

\begin{theorem}
\label{thm5}
The quadratic entropy $S_Q(\rho)=1-\Tr(\rho^2)$ satisfies the firm subadditivity condition and thus has all four properties given in Theorem~\ref{thm2}.
\end{theorem}

\begin{proof} Through simple algebra, rewrite $\chi_Q$ as:
\begin{align}
\label{eqn24}\chi_Q(\{p_j,\rho_j\})&=S_Q(\rho)-\sum_j p_j S_Q(\rho_j)\\
\label{eqn25}&=\sum_{j,j'>j}p_jp_{j'}D_{\HS}(\rho_j,\rho_{j'}).
\end{align}
By Lemma~\ref{thm4}, $\sum_{j,j'>j}p_jp_{j'}D_{\HS}(\dya{\psi_j},\dya{\psi_{j'}}) \geq \sum_{j,j'>j}p_jp_{j'}D_{\HS}[\EC(\dya{\psi_j}),\EC(\dya{\psi_{j'}})]$, meaning $\chi_Q(\{p_j,\dya{\psi_j}\}) \geq \chi_Q(\{p_j,\EC(\dya{\psi_j})\})$, which is condition (iv) of Theorem~\ref{thm2}, equivalent to the other three conditions of Theorem~\ref{thm2}.
\end{proof}

\begin{figure}[t]
\begin{center}
\includegraphics[width=8cm]{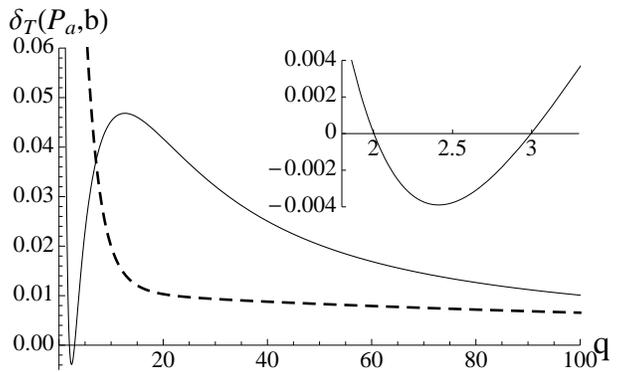}
\caption{%
 Discord $\dl_T(P_a,b)$ as a function of $q$ for two different choices (solid and dashed curves) of $\rho_{ab}$ and $P_a$. The inset zooms in on the region where the discord goes negative for the solid curve.\label{fgr1}}
\end{center}
\end{figure}

\section{Tsallis entropy}\label{sct6}

Since both $S$ and $S_Q$ are FSA and both belong to the Tsallis entropy family, it is natural to ask whether other Tsallis entropies are FSA, in particular for exponent $q\geq 1$ in \eqref{eqn1}, for which $S_T$ is subadditive. (Since $S_T$ is not subadditive over the range $0<q<1$, it cannot be firmly subadditive over this range.) We have found numerical examples that violate the FSA condition for $S_T$ over the range $2+\ep<q<3-\ep$ for $\ep=.005$. This is illustrated in Fig.~\ref{fgr1}, which plots the Tsallis discord $\dl_T(P_a,b)$ as a function of $q$ for two examples of $\rho_{ab}$ and $P_a$; negative discord indicates a violation of the FSA condition. This shows that some entropies, $S_T$ for $2+\ep<q<3-\ep$, are subadditive but not firmly subadditve. It is not known whether $S_T$ is FSA over the ranges $1<q<2$ and $q\geq 3$; our numerical searches in these ranges have not found any examples that violate the FSA condition.

\section{Generalized additivity}\label{sct7}

Thusfar we have discussed whether or not the discord $\dl_K$ is non-negative, consider the following remarks relevant to the conditions under which the discord is zero.

An entropy $S_K$ is \emph{additive} if
\begin{equation}
\label{eqn26}
S_K(\rho_a\ot\rho_b)=S_K(\rho_a)+S_K(\rho_b).
\end{equation}
\begin{proposition}
\label{thm5}
For some orthonormal basis $w=\{\ket{w_j}\}$ of $\HC_a$, the condition
\begin{equation}
\label{eqn27}
S_K(\sum_j p_j\dya{w_j}\ot \rho_{bj})=S_K(\rho_a)+\sum_j p_j S_K(\rho_{bj})
\end{equation}
generalizes (or strengthens) the additivity of $S_K$ in that \eqref{eqn27} implies \eqref{eqn26}; hence it may be called \emph{generalized additivity}. \openbox
\end{proposition}
This can be seen by writing $\rho_a\ot \rho_b=\sum_j p_j \dya{w_j}\ot\rho_b$, where $w$ is the eigenbasis of $\rho_a$, and applying \eqref{eqn27} to this state. Note that \eqref{eqn27} is equivalent to saying that $\dl_K(w,b)=0$ for $\rho_{ab}=\sum_j p_j\dya{w_j}\ot \rho_{bj}$. Since $S_T$ is not additive (except for $q=1$), it cannot satisfy \eqref{eqn27}. While $S_R$ is additive, it does \emph{not} satisfy \eqref{eqn27} (again with the exception of $q=1$, \cite{OllZur01}). Thus, of the entropies in \eqref{eqn1}, $S$ is unique in satisfying this generalized additivity condition.

\section{Conclusions}\label{sct8}

An information-theoretic condition, that the quantum mutual information upper-bounds the Holevo information, was proposed as a means to further classify quantum entropy functions. It is stronger than subadditivity, hence called the \emph{firm subadditivity}. It was proven that the quadratic entropy has this property, showing that an entropy can be firmly subadditive without being strongly subadditive, while some other Tsallis entropies with exponent between 2 and 3 do not have this property, showing that an entropy can be subadditive without being firmly subadditive. Several mathematically equivalent expressions of the firm subadditivity were given in Theorem~\ref{thm2}, providing bounds on, e.g. $\chi_Q$, the quadratic Holevo quantity, which we anticipate to be useful due to its simplicity and its connection to the Hilbert-Schmidt distance in \eqref{eqn25}. The definition of discord was extended to entropy functions other than von Neumann's $S$; however, it is not yet clear whether these other discords will be useful for quantifying the non-classicality of correlations, since only $S$ satisfies the generalized additivity condition in \eqref{eqn27}.

The author thanks Robert Griffiths and Shiang-Yong Looi for helpful suggestions, and acknowledges support from the Office of Naval Research.

\bibliographystyle{unsrt}
\bibliography{FirmSA}

\end{document}